\documentclass[journal,twocolumn]{IEEEtran}

\usepackage{graphicx,amsmath,amssymb,amsfonts,cite}
\usepackage{subfigure}
\usepackage{algorithm}
\usepackage{algorithmic}
\usepackage{float}
\usepackage{amsthm}
\usepackage[english]{babel}
\usepackage[centerfoot]{pageno}
\usepackage{ulem}
\usepackage{bm}
\usepackage{subfigure}
\usepackage{multirow}
\usepackage{array}
\DeclareRobustCommand{\stirling}{\genfrac\{\}{0pt}{}}
\ifCLASSINFOpdf
\else
\fi
\begin{document}
%
\title{
Success Probability of Grant-Free Random Access with Massive MIMO}

\author{Jie~Ding,~
        Daiming~Qu,~
        Hao Jiang,~
        and~Tao~Jiang,~\IEEEmembership{Senior Member,~IEEE}
\thanks{Jie Ding, Daiming Qu, Hao Jiang, and Tao Jiang are with
School of Electronic Information
and Communications, Huazhong University of Science and Technology,
Wuhan, 430074, China.}
\thanks{This work was supported in part by the National Natural Science Foundation of China funded project under grant number 61701186 and the China Postdoctoral Science Foundation funded project under grant number 2017M612458.}}

\maketitle

\begin{abstract}
Massive MIMO opens up new avenues for enabling highly efficient random access (RA) by
offering abundance of spatial degrees of freedom.
In this paper, we investigate the grant-free RA with massive MIMO and derive the analytic expressions of success probability of the grant-free RA for conjugate beamforming and zero-forcing beamforming techniques. With the derived analytic expressions, we further shed light on the impact of system parameters on the success probability. Simulation results verify the accuracy of the analyses. It is conf\/irmed that the grant-free RA with massive MIMO is an attractive RA technique with low signaling overhead that could simultaneously accommodate a number of RA users, which is multiple times the number of RA channels, with close-to-one success probability. In addition, when the number of antennas in massive MIMO is suf\/f\/iciently large, we show that the number of orthogonal preambles would dominate the success probability.
\end{abstract}

\begin{IEEEkeywords}
Success probability, Grant-free, Random access, Massive MIMO, M2M.
\end{IEEEkeywords}

%
\IEEEpeerreviewmaketitle

\section{Introduction}
 Future mobile communication systems not only envision enhancing the traditional mobile broadband use case, but also aim to meet the requirements of new emerging use cases, such as Internet of Things (IoT) \cite{1}\cite{2}. As an enabler of the IoT, machine-to-machine (M$2$M) communications have attracted considerable attention from academia and industries. In M$2$M, the number of random access (RA) user equipments (UEs) is enormous and their data packets are usually short and sporadic in nature. As a result, fulfilling the demand of massive access with low signaling overhead and access delay is a key technological issue in future wireless communications \cite{3}.

The legacy request-grant RA procedure in long term evolution (LTE) was only designed to provide reliable access to a small number of UEs with long packets to transmit \cite{4}. To support M$2$M communications, several modifications and improvements have been proposed \cite{7, 8, 9}.
Additionally, a new narrowband IoT (NB-IoT) technology, based on LTE, has been standardized by 3GPP to this end. In \cite{34, 35, 36}, design and optimization of RA in NB-IoT have been presented, shedding light on the potential of NB-IoT toward supporting M$2$M communications. Nevertheless, since the RA of NB-IoT is a request-grant protocol based on slotted-ALOHA and very
limited wireless resources are provided for NB-IoT RA UEs, it is unable to support massive access as required by the M$2$M, where low signaling overhead and access delay are essential.

Recently, massive multiple-input multiple-output (MIMO), has been identified as a promising technology to mitigate the wireless resource scarcity and handle the rapid growth of data traffic \cite{14, 15, 30}, which opens up new avenues for massive access by offering abundance of spatial degrees of freedom \cite{31}. Several works were devoted to improving the legacy request-grant RA procedure \cite{16,17,18} by taking advantage of high spatial resolution and channel hardening of massive MIMO \cite{15}.
These works validated the effectiveness of massive MIMO in resolving access collision and enhancing RA capacity. However, considering small-sized packets generated by IoT applications, the request-grant RA procedure brings in relatively long waiting time before data transmissions for RA UEs. Moreover,
since the channel resources reserved for request and grant signaling are not utilized as efficiently as the data channel that takes full advantage of massive MIMO, the request-grant RA is not an efficient approach in the case of massive MIMO.
%

To effectively manage M$2$M communications at low signaling overhead and access delay, grant-free RA (also known as one-stage RA) with massive MIMO is a compelling alternative. In the grant-free RA, request-grant procedure in the legacy RA is omitted and RA UEs contend (i.e., perform RA)
with their uplink payloads directly by transmitting preamble along with data \cite{19}. As a result, signaling overhead and access delay are minimized, and the radio resources reserved for the request-grant procedure could be unleashed for accommodating more RA UEs. With all the benefits manifested in \cite{16,17,18}, features of massive MIMO could be exploited to effectively accommodating multiple access in RA over the same channel.
Therefore, the grant-free RA with massive MIMO exhibits potential advantages towards addressing RA issues for future wireless communications.
However, to the best of the authors knowledge, this paper is the first one to investigate the performance of grant-free RA with massive MIMO. In \cite{20}, a joint pilot assignment and data transmission protocol was proposed to support massive intermittent transmissions, which relies on pilot-hopping patterns across multiple transmission slots. Since this protocol assumed that each RA UE is associated to a unique and predefined pilot-hopping pattern and the BS
knows in advance the pilot-hopping patterns of all RA UEs, it is not a genuine grant-free RA protocol with which RA UEs compete for the channel access by transmitting
preambles randomly chosen from a preamble pool.




In this regard, the grant-free RA with massive MIMO is investigated in this paper to provide insights into the design of RA protocols for future wireless communications. Specifically, we consider a grant-free RA scenario that a large number of RA UEs contend for access to limited channel resources by directly transmitting preamble along with data.
Success probability, as the performance metric, is used in this paper to evaluate the effectiveness of grant-free RA with massive MIMO.

The novelty and contribution of this paper are summarized as follows.
\begin{itemize}
  \item We propose the idea of grant-free RA with massive MIMO and derive approximate expressions of the success probability of the grant-free RA for the cases of conjugate beamforming (CB) and zero-forcing beamforming (ZFB), respectively.

  \item Taking into consideration that the number of antennas $M$ in massive MIMO is usually suf\/f\/iciently large, it is found that the number of orthogonal preambles $P$ would dominate the success probability. 

  \item Simulation results show that our analyses are accurate and the grant-free RA with massive MIMO is able to support $N_a$ simultaneous grant-free access over $C$ RA channels with close-to-one success probability, where $N_a$ is multiple times $C$. A great MIMO gain in terms of $\eta$ could be achieved for the grant-free RA compared to its single-antenna counterpart, where $\eta=N_a/C$ reflects the channel reuse efficiency.

  \item It is demonstrated that the grant-free RA with massive MIMO achieves a close performance to the one with even user distribution (EUD) over channels. This is an important merit of the grant-free RA with massive MIMO, considering the fact that the EUD is desirable but unattainable in the grant-free RA.
\end{itemize}

The remainder of this paper is organized as follows. In Section II,
the grant-free RA with massive MIMO is briefly described. In Section III, analyses and derivations on the success probability of grant-free RA with massive MIMO are detailed.
Simulation results are
presented in Section IV and the work is concluded in Section V.

\emph{Notations: }Boldface lower and upper case symbols represent
vectors and matrices, respectively. $\mathbf{I}_n$ is the $n \times n$ identity matrix. The trace, conjugate, transpose, and complex conjugate transpose
operators are denoted by $\mathrm{tr}(\cdot)$, $(\cdot)^{*}$, $(\cdot)^{\mathrm{T}}$ and $(\cdot)^{\mathrm{H}}$. $\mathbb{E}[\cdot]$ denotes the expectation operator. $\|\cdot\|$ denotes the Euclidean norm and $[\mathbf{G}]_{ij}$ denotes the entry of matrix $\mathbf{G}$ on the $i$th row and $j$th column.
$\mathbf{x}\sim \mathcal{CN}(0,\mathbf{\Sigma})$ indicates that x is a circularly symmetric
complex Gaussian (CSCG) random vector with zero-mean and
covariance matrix $\mathbf{\Sigma}$. $\mathcal{B}(r,n,p)= {\binom{n}{r}}p^r(1-p)^{n-r}$ is the probability
mass function of a binomial distribution with parameters $r$, $n$, and $p$, where ${\binom{n}{r}}=\frac{n!}{r!(n-r)!}$ is the binomial coefficient.

\section{Grant-Free Random Access with Massive MIMO}
In this section, we firstly outline the system model with massive MIMO considered in this paper. Then, the procedure of grant-free RA with massive MIMO is introduced.

In Fig. \ref{fig1}, a single-cell massive MIMO system is depicted, where the base station (BS) is configured with $M$ active antenna elements and single-antenna RA UEs are uniformly distributed throughout the cell. With spatial
reuse capability of massive MIMO, the BS is able to support multiple RA UEs simultaneously over a same channel. In this paper, we assume that $N_a$ RA UEs are active to perform RA in a RA slot. The RA slot herein is the wireless radio resources dedicated to the grant-free RA, which consists of $C$ channels over frequency, as shown in Fig. \ref{fig2}.
\begin{figure}[!h]
\centering
\includegraphics[width=3.3in]{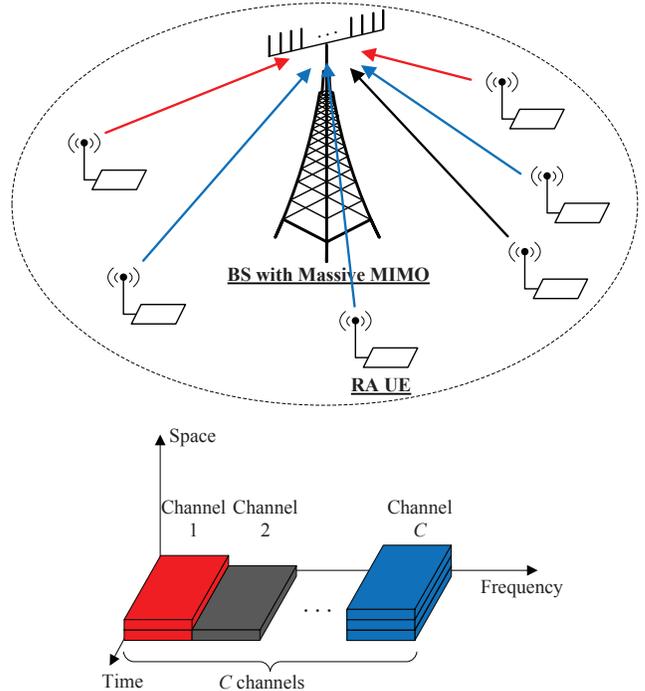}
\caption{System model with massive MIMO.} \label{fig1}
\end{figure}

Since $M$ is sufficiently large in massive MIMO, favorable propagation (FP) can be approximately achieved, which means that RA UEs' channel vectors are approximately orthogonal. The feature of FP enables spatial multiplexing of multiple RA UEs over a same channel \cite{14}\cite{15}. Specifically, simple linear processing, such as CB and ZFB, could be applied at the BS, to discriminate the signal transmitted by each RA UE from the signals of other RA UEs.

\begin{figure}[!h]
\centering
\includegraphics[width=2.6in]{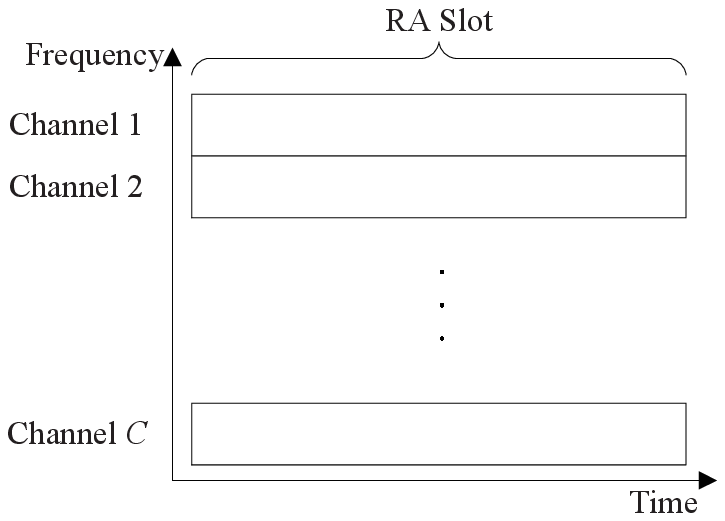}
\caption{Time-frequency resources in a RA slot.} \label{fig2}
\end{figure}

In Fig. \ref{fig3}, the procedure of grant-free RA with massive MIMO is briefly described. Specifically, once a RA UE decides to transmit at the RA slot, it randomly selects one of the $C$ RA channels. Over the selected channel, the RA UE transmits a RA preamble followed by its data. The RA preamble is randomly chosen out of a RA preamble pool, which is used by the BS for preamble detection and channel estimation. We assume that there are $P$ orthogonal RA preambles available in the pool. If the chosen preamble by a RA UE is different from the ones by other RA UEs over the same channel, the RA UE could be detected and its channel response could be estimated with adequate accuracy at the BS. Otherwise, if multiple RA UEs choose the same preamble over the same channel, preamble collision occurs and we assume that all these RA UEs would not be detected and their channel estimations would be failed at the BS.

Information about the available RA preambles is periodically broadcasted by the BS. For RA UEs with only small-sized packet transmissions, they are able to achieve reduced signaling overhead and transmission delay as well as effective power saving with the grant-free RA.
At the BS, after preamble detection and channel estimation, receive beamforming is then used for data recovering.

\begin{figure}[!h]
\centering
\includegraphics[width=2.6in]{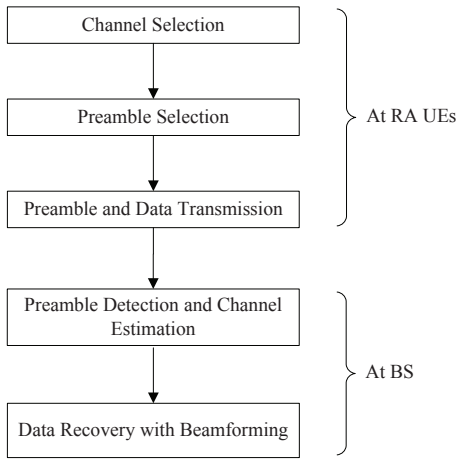}
\caption{Procedure of grant-free RA with massive MIMO.} \label{fig3}
\end{figure}

 By utilizing spatial multiplexing offered by the excess spatial degrees of freedom, massive MIMO is able to serve multiple RA UEs over a same channel, which makes it possible for the grant-free RA with massive MIMO to meet the need of massive access.
In Section III.A, key factors influencing the performance of multiple simultaneous access in the grant-free RA with massive MIMO are discussed.

\section{Success Probability of Grant-Free RA with Massive MIMO }
To evaluate the performance of the grant-free RA with massive MIMO, we use success probability as the performance metric. In this paper, the success probability is defined as the probability of no preamble collision and $\gamma \geq \gamma_{\mathrm{Th}}$ for an arbitrary RA UE, where $\gamma$ is the received signal to interference and noise ratio (SINR) after beamforming at the BS and $\gamma_{\mathrm{Th}}$ is a given SINR threshold. 

In this section, we firstly analyze key factors influencing the performance of grant-free RA with massive MIMO and then derive the success probabilities of the grant-free RA in massive MIMO for CB and ZFB techniques. Lastly, a proposition is provided to shed light on how $M$ and $P$ affect the success probability of the grant-free RA with massive MIMO.

\subsection{Factors Influencing Success Probability}

In order to successfully decode the data of a RA UE in the grant-free RA, it is essential for the BS to have accurate channel response of the RA UE. The BS acquires the channel response from the preamble transmitted by the RA UE. In the case of multiple simultaneous access over a same channel, if other RA UEs select the same preamble with the RA UE, the preamble collision occurs. The BS is thus unable to acquire the correct channel response for the RA UE and the data transmission would be failed. Moreover, the incorrect channel responses lead to an incorrect beamforming pattern (especially for ZFB), which would bring in multiuser interference to other RA UEs and degrade the beamforming performance as a result.
On the other hand, even that the BS acquires good channel
estimations for all the RA UEs, data transmission of the RA UEs would also not be successful if the multiuser interference is sufficiently strong such that $\gamma < \gamma_{\mathrm{Th}}$ .


In summary, the key factors influencing the success probability can be outlined as follows
\begin{enumerate}
   \item Preamble collision. 
      Only RA UEs without experiencing preamble collisions have the chance to get their data recovered by the BS.
  \item Noise and multiuser interference. Multiuser interference in beamforming would result in loss of $\gamma$. Data transmission of a RA UE is considered successful only when its $\gamma$ is greater than $\gamma_{\mathrm{Th}}$. 

\end{enumerate}

In the sequel, the mathematical expressions of success probability are derived by taking the above influencing factors into account. To make the derivations trackable, a block independent Rayleigh fading propagation model is considered, where the propagation channels are assumed constant within the RA slot.
The channel response vector between an arbitrary RA UE and the BS is modelled by $\mathbf{g}=\sqrt{\ell}\mathbf{h} \in \mathbb{C}^{M}$, where $\ell$ denotes the large scale fading coefficient between RA UE and BS, and $\mathbf{h} \sim \mathcal{CN}(0,\mathbf{I}_{M})$ stands for the small scale fading vector between RA UE and BS. In the derivations, we ignore the impact of noise on the channel estimation. Moreover, perfect power control is assumed so that all RA UEs have the same expected receive power at the BS.

\subsection{Success Probability of Grant-Free RA with Conjugate Beamforming}
In this subsection, we derive the success probability of the grant-free RA with CB.
Without loss of generality, we take the $1$st RA UE as an example to specify the theoretical derivations.
Then, the success probability is expressed as
\begin{align}\label{eq2}
{P}_\mathrm{CB}=\sum_{K=0}^{N_a-1}\bar{P}(K)\tilde{P}_{\mathrm{CB}}(K){P}(\gamma^1_{\mathrm{CB}} \geq \gamma_{\mathrm{Th}}|K),
\end{align}
where $\gamma^1_{\mathrm{CB}}$ is the SINR corresponding to the $1$st RA UE after CB at the BS and the subscript $\mathrm{CB}$ indicates that the CB is utilized. $N_a$ is the number of RA UEs for grant-free RA in a RA slot.

$\bar{P}(K)$ represents the probability that $K$ other RA UEs select the same channel with the $1$st RA UE, which is given by
\begin{align}\label{eq3}
\bar{P}(K)=\mathcal{B}(K,N_a-1,\frac{1}{C}).
\end{align}

$\tilde{P}_{\mathrm{CB}}(K)$ represents the probability that no preamble collisions occur between the $K$ RA UEs and the $1$st RA UE, which is given by
\begin{align}\label{eq4}
\tilde{P}_{\mathrm{CB}}(K)=(1-\frac{1}{P})^K.
\end{align}

${P}(\gamma^1_{\mathrm{CB}} \geq \gamma_{\mathrm{Th}}|K) $ is the probability of $\gamma^1_{\mathrm{CB}} \geq \gamma_{\mathrm{Th}}$, given that 1) the channel response of the $1$st RA UE is available to the BS; 2) $K$ other RA UEs share the same channel resource with the $1$st RA UE for data transmissions.

Under the assumption of perfect power control at each RA UE, let $P_{\mathrm{R}}$ denote the expected receive power from each RA UE at each BS antenna. The received uplink signal vector $\mathbf{r} \in \mathbb{C}^{M}$ at the BS is thus written as
\begin{align}\label{eq5}
\mathbf{r}
=\sum_{i=1}^{K+1}\sqrt{P_{\mathrm{R}}}\mathbf{{h}}_{i}x_{i}+\mathbf{{n}},
\end{align}
where the index set of $K$ RA UEs is assumed to be $\mathcal{I}=\{2,3,\ldots, K+1\}$. $\mathbf{{h}}_{i}$ is the small scale fading vector between the $i$th RA UE and the BS. $x_{i} $ is data symbol transmitted by the $i$th RA UE and $\mathbb{E}[|x_{i}|^2]=1$. $\mathbf{{n}}\sim \mathcal{CN}(0,\sigma_\mathrm{{n}}^2\mathbf{I}_{M})$ is a vector of the additive white Gaussian noise (AWGN). We denote the uplink SNR at the BS corresponding to each RA UE by $\rho_{\mathrm{R}}\triangleq P_{\mathrm{R}}/\sigma_\mathrm{{n}}^2$.

With $\mathbf{r}$, the BS recovers the $1$st RA UE's data symbols by CB based on
its channel response. Then, the recovered signal corresponding to the $1$st RA UE after CB is given as
\begin{align}\label{eq6}
y_{1}&=\mathbf{b}^{\mathrm{T}}_1\mathbf{r} \nonumber\\
&=\sqrt{P_{\mathrm{R}}}\mathbf{h}^{\mathrm{H}}_{1}\mathbf{h}_{1}x_{1}+\sum_{i=2}^{K+1}\sqrt{P_{\mathrm{R}}}\mathbf{h}^{\mathrm{H}}_{1}\mathbf{h}_{i}x_{i}+\mathbf{h}^{\mathrm{H}}_{1}\mathbf{{n}},
\end{align}
where $\mathbf{b}^{\mathrm{T}}_1=\mathbf{h}^{\mathrm{H}}_{1}$ refers to the receive conjugate beamformer for the $1$st RA UE.
The SINR of the $1$st RA UE is therefore calculated as
\begin{align}\label{eq7}
\gamma^1_{\mathrm{CB}}=\frac{P_{\mathrm{R}}|\mathbf{h}^{\mathrm{H}}_{1}\mathbf{h}_{1}|^2}{|\mathbf{h}^{\mathrm{H}}_{1}\mathbf{{n}}|^2+P_{\mathrm{R}}\sum\limits_{i=2}^{K+1}|\mathbf{h}^{\mathrm{H}}_{1}\mathbf{h}_{i}|^2}.
\end{align}
In massive MIMO, $M$ is assumed large. By using the strong law of large numbers, we have $\frac{\mathbf{h}^{\mathrm{H}}_{1}\mathbf{h}_{1}}{M} \xrightarrow{M\rightarrow \infty} 1$. Thus, (\ref{eq5}) can be simplified as follows
\begin{align}\label{eq8}
\gamma^1_{\mathrm{CB}}=\frac{\rho_{\mathrm{R}}M}{1+\rho_{\mathrm{R}}Y_K},
\end{align}
where $Y_K=\frac{1}{M}\sum\limits_{i=2}^{K+1}|\mathbf{h}^{\mathrm{H}}_{1}\mathbf{h}_{i}|^2$. It is known that $\frac{1}{M}\mathbf{h}^{\mathrm{H}}_{1}\mathbf{h}_{i}$ converges
to the standard normal distribution \cite{33}. Then, $\frac{1}{M}|\mathbf{h}^{\mathrm{H}}_{1}\mathbf{h}_{i}|^2$ converges to the Gamma distribution $\phi(y;1,1)$. From Corollary $1$ of \cite{33}, the probability density function (PDF) of ${Y}_K$ has the following approximation:
\begin{align}\label{eq9}
&f_{{Y}_K}(y) \nonumber\\
\approx& \beta\eta^{-K+1}\left[e^{-\beta y}-e^{-\frac{\sqrt{M}}{\sqrt{M}-1}y}\sum_{n=0}^{K-2}\left(\frac{\sqrt{M}}{\sqrt{M}-1}\eta\right)^n\frac{y^n}{n!} \right],
\end{align}
where $\beta=\frac{\sqrt{M}}{\sqrt{M}+K-1}$ and $\eta=\frac{K}{\sqrt{M}+K-1}$.

By using (\ref{eq8}), ${P}(\gamma^1_{\mathrm{CB}} \geq \gamma_{\mathrm{Th}}|K) $ is obtained as
\begin{align}\label{eq10}
&{P}(\gamma^1_{\mathrm{CB}} \geq \gamma_{\mathrm{Th}}|K) \nonumber\\
=&{P}(\frac{\rho_{\mathrm{R}}M}{1+\rho_{\mathrm{R}}Y_K} \geq \gamma_{\mathrm{Th}}) \nonumber\\
=& \left\{
       \begin{array}{ll}
         {P}(Y_K \leq \frac{M}{\gamma_{\mathrm{Th}}}-\frac{1}{\rho_{\mathrm{R}}}), & \hbox{if $\gamma_{\mathrm{Th}} \leq M\rho_{\mathrm{R}}$;} \\
         0, & \hbox{ otherwise.}
       \end{array}
     \right.
\end{align}

In a typical massive MIMO setup, the condition of $\gamma_{\mathrm{Th}} \leq M\rho_{\mathrm{R}}$ is usually satisfied. ${P}(\gamma^1_{\mathrm{CB}} \geq \gamma_{\mathrm{Th}}|K) $ is thus approximated as the following by using (\ref{eq9})
\begin{align}\label{eq11}
&{P}(\gamma^1_{\mathrm{CB}} \geq \gamma_{\mathrm{Th}}|K) \nonumber\\
=&\int_{0}^{\frac{M}{\gamma_{\mathrm{Th}}}-\frac{1}{\rho_{\mathrm{R}}}}f_{{Y}_K}(y)\mathrm{d}y \nonumber\\
=&1-\eta^{-K+1}e^{-\beta{\Lambda}} \nonumber\\
+&(1-\eta)\sum_{n=0}^{K-2}\frac{1}{n!}\eta^{n-K+1}\Gamma(n+1,\frac{\sqrt{M}}{\sqrt{M}-1}{\Lambda}),
\end{align}
where ${\Lambda}=\frac{M}{\gamma_{\mathrm{Th}}}-\frac{1}{\rho_{\mathrm{R}}}$. $\Gamma(s,x)=\int_{x}^{\infty}t^{s-1}e^{-t}\mathrm{d}t $ is the
upper incomplete Gamma function.

Substituting (\ref{eq3}), (\ref{eq4}) and (\ref{eq11}) into (\ref{eq2}), the approximate expression of the success probability of the grant-free RA with CB is obtained.

\subsection{Success Probability of Grant-Free RA with Zero-Forcing Beamforming}
Similar to the case of CB, the success probability of the grant-free RA with ZFB is expressed as
\begin{align}\label{eq12}
&{P}_\mathrm{ZF} \nonumber\\
=&\sum_{K=0}^{N_a-1}\bar{P}(K)\sum_{S=1}^{\min\{P-1, K\}}\tilde{P}_{\mathrm{ZF}}(S|K){P}(\gamma^1_{\mathrm{ZF}} \geq \gamma_{\mathrm{Th}}|K,S),
\end{align}
where $\gamma^1_{\mathrm{ZF}}$ is the SINR corresponding to the $1$st RA UE after ZFB at the BS and the subscript $\mathrm{ZF}$ indicates that the ZFB is utilized.
$\bar{P}(K)$ is defined in (\ref{eq3}), referring to the probability that $K$ other RA UEs select the same channel with the $1$st RA UE.

$\tilde{P}_{\mathrm{ZF}}(S|K)$ refers to the probability that $S$ preambles in total are selected by the $K$ cochannel RA UEs and no preamble collisions occur between the $K$ RA UEs and the $1$st RA UE, given that $K$ other RA UEs select the same channel with the $1$st RA UE, which is given by
\begin{align}\label{eq13}
\tilde{P}_{\mathrm{ZF}}(S|K)=\frac{{\binom{P-1}{S}}S!\stirling{K}{S}}{P^K},
\end{align}
where $\stirling{K}{S}=\frac{1}{S!}\sum\limits_{j=0}^{S}(-1)^{S-j}\binom{S}{j}j^K$ is the Stirling numbers of the second kind \cite{24}. It is the number of ways to partition a set of $K$ objects into $S$ non-empty subsets. In short words, it is the number of ways to distribute $K$ distinguishable elements into $S$ indistinguishable receptacles with no receptacle empty.

${P}(\gamma^1_{\mathrm{ZF}} \geq \gamma_{\mathrm{Th}}|K,S) $ is the probability of $\gamma^1_{\mathrm{ZF}} \geq \gamma_{\mathrm{Th}}$, given that 1) $K$ other RA UEs share the same channel resource with the $1$st RA UE for data transmissions; 2) $S$ preambles in total are selected by the $K$ RA UEs, which are different from the one selected by the $1$st RA UE. The mathematical expression of ${P}(\gamma^1_{\mathrm{ZF}} \geq \gamma_{\mathrm{Th}}|K,S) $ is analyzed and derived in the followings.

Unlike the case of CB, ZFB utilizes all the estimated channel information to decode the data. As $S$ preambles in total are selected by the $K$ RA UEs, the BS would obtain $S$ channel estimates by detecting the $S$ selected preambles, where $S \leq K$.
When $S<K$, some of the $S$ channel estimates would be the superposition of multiple RA UEs' channel responses.
As a result, the ZFB with the incorrect channel estimates would impose interference to the targeted UE, which has to be taken into consideration in the derivations.

In this regard, we define $\mathcal{W}_s: s\in \mathcal{S}=\{1,2,\ldots,S\}$ as the nonempty subset of RA UEs that select the $s$th preamble in the selected $S$ preambles. $\mathcal{W}_s \subset \mathcal{I}$, 
where $\mathcal{I}$ is the index set of the $K$ RA UEs defined in (\ref{eq5}).
We also define that $w_s$ is an arbitrary element in subset $\mathcal{W}_s$ and $\mathcal{W}_s\setminus{w_s}$ refers to the subset $\mathcal{W}_s$ excluding the element $w_s$.
In addition, we define $\mathbf{a}_s$, for $s\in \mathcal{S}$, as the channel response estimated via the detection of the $s$th preamble, where $\mathbf{a}_s=\sum\limits_{j\in \mathcal{W}_s}\mathbf{h}_j$ as the impact of noise on the channel estimation is ignored herein.

With the above definitions and $\mathbf{r}$ in (\ref{eq5}), the recovered signal vector $\mathbf{y} \in \mathbb{C}^{S+1} $ after ZFB is given as
\begin{align}\label{eq14}
\mathbf{y}&=\mathbf{B}\mathbf{r}=(\mathbf{A}^{\mathrm{H}}\mathbf{A})^{-1}\mathbf{A}^\mathrm{H}\mathbf{r},
\end{align}
where $\mathbf{A}=[\mathbf{h}_1,\mathbf{a}_1,\mathbf{a}_2,\ldots,\mathbf{a}_S] \in \mathbb{C}^{M \times (S+1) } $
and $\mathbf{B}=(\mathbf{A}^{\mathrm{H}}\mathbf{A})^{-1}\mathbf{A}^\mathrm{H} \in \mathbb{C}^{(S+1)\times M }$ refers to the receive ZF beamformer.

To make the derivations tractable, we apply the following transformation to $\mathbf{r}$:
\begin{align}\label{eq15}
\mathbf{r}&=\sum_{i=1}^{K+1}\sqrt{P_{\mathrm{R}}}\mathbf{{h}}_{i}x_{i}+\mathbf{{n}} \nonumber\\
&=\sqrt{P_{\mathrm{R}}}\mathbf{A}\hat{\mathbf{x}}+\underbrace{\sum\limits_{j\in \{\mathcal{W}_s\setminus{w_s}: s\in \mathcal{S}\}}\sqrt{P_{\mathrm{R}}}\mathbf{h}_j(x_j-x_{w_s})}_{\mathrm{Sum}~\mathrm{of}~K-S~\mathrm{terms}}+\mathbf{{n}},
\end{align}
where $\hat{\mathbf{x}}=[x_1, x_{w_1},x_{w_2},\ldots,x_{w_S}]^{\mathrm{T}}\in \mathbb{C}^{S+1 }$. Then, two cases are discussed as follows:

\subsubsection{Case of $K>S$}

Substituting (\ref{eq15}) into (\ref{eq14}), we have that
\begin{align}\label{eq16}
\mathbf{y}&=\sqrt{P_{\mathrm{R}}}\hat{\mathbf{x}}+\sum\limits_{j\in \{\mathcal{W}_s\setminus{w_s}: s\in \mathcal{S}\}}\sqrt{P_{\mathrm{R}}}\mathbf{B}\mathbf{h}_j(x_j-x_{w_s})+\mathbf{B}\mathbf{{n}}.
\end{align}
Correspondingly, the recovered signal for the $1$st RA UE at the BS is derived as
\begin{align}\label{eq17}
y_{1}&=\sqrt{P_{\mathrm{R}}}x_1+\sum\limits_{j\in \{\mathcal{W}_s\setminus{w_s}: s\in \mathcal{S}\}}\sqrt{P_{\mathrm{R}}}\mathbf{b}^{\mathrm{T}}_{1}\mathbf{h}_jx_{w_s, j}+\mathbf{b}^{\mathrm{T}}_{1}\mathbf{{n}},
\end{align}
where $\mathbf{b}^{\mathrm{T}}_{1}$ denotes the $1$st row of matrix $\mathbf{B}$. $x_{w_s, j}=x_j-x_{w_s}$ and $\mathbb{E}[|x_{w_s, j}|^2]=2$ since the data symbols are independent to each other.

From (\ref{eq17}), the SINR of the $1$st RA UE is written as
\begin{align}\label{eq18}
\gamma^1_{\mathrm{ZF}}&=\frac{\rho_{\mathrm{R}}}{\rho_{\mathrm{R}}\Big|\sum\limits_{j\in \{\mathcal{W}_s\setminus{w_s}: s\in \mathcal{S}\}}\sqrt{2}\mathbf{b}^{\mathrm{T}}_{1}\mathbf{h}_j\Big|^2  +\|\mathbf{b}^{\mathrm{T}}_{1}\|^2} \nonumber\\
&= \frac{\rho_{\mathrm{R}}U_1}{\rho_{\mathrm{R}}Z_{K-S}+1},
\end{align}
where
\begin{align*}
U_1=\|\mathbf{b}^{\mathrm{T}}_{1}\|^{-2}.
\end{align*}
And
\begin{align*}
Z_{K-S}=\big|\underbrace{\sum\limits_{j\in \{\mathcal{W}_s\setminus{w_s}: s\in \mathcal{S}\}}\frac{\sqrt{2}\mathbf{b}^{\mathrm{T}}_{1}}{\|\mathbf{b}^{\mathrm{T}}_{1}\|}\mathbf{h}_j}_{\mathrm{Sum}~\mathrm{of}~K-S~\mathrm{terms}}\big|^2.
\end{align*}

Since $\|\mathbf{b}^{\mathrm{T}}_{1}\|^{2}=[(\mathbf{A}^{\mathrm{H}}\mathbf{A})^{-1}]_{11}$, $U_1$ follows an Erlang distribution
with shape parameter $M-S$ and scale parameter $1$ \cite{25}. Therefore, the PDF of $U_1$ is given by
\begin{align}\label{eq19}
f_{{U}_1}(u)=\frac{e^{-u}}{(M-S-1)!}u^{M-S-1},~~~u>0.
\end{align}

On the other hand, since it is difficult to derive the exact distribution of $Z_{K-S}$, an approximation is considered herein. Specifically, by ignoring the correlations between the $K-S$ terms in the sum, we approximate $Z_{K-S}$ as
\begin{align}\label{eq20}
Z_{K-S}
&\approx \sum\limits_{j\in \{\mathcal{W}_s\setminus{w_s}: s\in \mathcal{S}\}}\big|\frac{\sqrt{2}\mathbf{b}^{\mathrm{T}}_{1}\mathbf{h}_j}{\|\mathbf{b}^{\mathrm{T}}_{1}\|}\big|^2  \nonumber\\
&=\sum\limits_{j\in \{\mathcal{W}_s\setminus{w_s}: s\in \mathcal{S}\}}\big|K_j\big|^2,
\end{align}
where $K_j=\frac{\sqrt{2}\mathbf{b}^{\mathrm{T}}_{1}\mathbf{h}_j}{\|\mathbf{b}^{\mathrm{T}}_{1}\|}$. It is proved that $K_j$ follows standard complex normal distribution \cite{28}. Therefore, $Z_{K-S}$ is approximated as the sum of $K-S$ statistically independent and identically distributed exponential random variables. As a result, the PDF of $Z_{K-S}$ is given by \cite{26}\cite{27}
\begin{align}\label{eq21}
f_{{Z}_{K-S}}(z)\approx\frac{1}{(K-S-1)!}z^{K-S-1}e^{-z},~~~z>0.
\end{align}

With (\ref{eq18}), (\ref{eq19}), and (\ref{eq21}), we obtain that
\begin{align}\label{eq22}
&{P}(\gamma^1_{\mathrm{ZF}} \geq \gamma_{\mathrm{Th}}|K,S) \nonumber\\
\approx&e^{-\frac{\gamma_{\mathrm{Th}}}{\rho_{\mathrm{R}}}}\sum_{p=0}^{M-S-1}\sum_{q=0}^{p}{\binom{p}{q}}\frac{\gamma^p_{\mathrm{Th}}}{p!}
\frac{\rho^{q-p}_{\mathrm{R}}}{(K-S-1)!}\frac{\Gamma(K-S+q,0)}{(1+\gamma_{\mathrm{Th}})^{K-S+q}}.
\end{align}

\subsubsection{Case of $K=S$}

In the case of $K=S$, no preamble collisions occur among the $K$ RA UEs, i.e., their channel estimations would be successful, and the multiuser interference could be suppressed effectively with ZFB. The recovered signal for the $1$st RA UE at the BS in (\ref{eq17}) is hence simplified as
\begin{align}\label{eq23}
y_{1}&=\sqrt{P_{\mathrm{R}}}x_1+\mathbf{b}^{\mathrm{T}}_{1}\mathbf{{n}}.
\end{align}

Accordingly, the SINR of the $1$st RA UE is written as
\begin{align}\label{eq24}
\gamma^1_{\mathrm{ZF}}=\rho_{\mathrm{R}}U_1.
\end{align}
With (\ref{eq19}), $\mathbb{P}(\gamma^1_{\mathrm{ZF}} \geq \gamma_{\mathrm{Th}}|K,S)$ in the case of $K=S$ is therefore given by
\begin{align}\label{eq25}
&{P}(\gamma^1_{\mathrm{ZF}} \geq \gamma_{\mathrm{Th}}|K,S) \nonumber\\
=&{P}(\gamma^1_{\mathrm{ZF}} \geq \gamma_{\mathrm{Th}}|K,K)  \nonumber\\ =&e^{-\frac{\gamma_{\mathrm{Th}}}{\rho_{\mathrm{R}}}}\sum_{p=0}^{M-K-1}\frac{1}{p!}(\frac{\gamma_{\mathrm{Th}}}{\rho_{\mathrm{R}}})^p.
\end{align}

Substituting (\ref{eq3}), (\ref{eq13}), (\ref{eq22}) and (\ref{eq25}) into (\ref{eq12}), the approximate expression of the success probability of the grant-free RA with ZFB is obtained.

\subsection{Analysis under Assumption of Even User Distribution}
In this subsection, we simplify the analytic expressions of the success probability under assumption of even user distribution (EUD) over channels to shed light on how $M$ and $P$ impact on the success probability of the grant-free RA with massive MIMO.
Herein, the EUD assumes a genie user distribution management such that the number of RA UEs distributed on each RA channel is exact $N_a/C$. Please note that, the number of RA UEs on each RA channel is random in practice, therefore the EUD is desirable but unattainable in the grant-free RA.
The reasons of making the assumption for the analysis lie in:
\begin{enumerate}
  \item The assumption of the EUD eliminates the effects of random user distribution over channels, which simplifies the analytic expressions in (\ref{eq2}) and (\ref{eq12}), thus making it more straightforward to understand how $M$ and $P$ impact on the success probability;
  \item Although the assumption of the EUD is impractical, the number of RA UEs distributed on each RA channel with random user distribution is close to $N_a/C$ when $N_a\gg C$. Thus, it is a reasonable approximation to the real case with random user distribution.
  \item The EUD provides a performance upper bound for the grant-free RA. By comparing to the upper bound, we can evaluate the performance gap to this desirable performance.

\end{enumerate}

With the EUD assumption, one proposition is hereby introduced.

\newtheorem{Prop}{Proposition}
\begin{Prop}
In the grant-free RA with EUD, when $M$ is sufficiently large, the success probability approaches to $(1-\frac{1}{P})^{\frac{N_a}{C}-1}$ for both CB and ZFB.
\end{Prop}

\begin{proof}
We only provide the proof for CB herein. The proof for ZFB is similar and omitted due
to space constraints.

With CB and EUD, the success probability in (\ref{eq2}) could be modified as
\begin{align}\label{eq26}
{P}_\mathrm{CB}=(1-\frac{1}{P})^{\frac{N_a}{C}-1}{P}\big(\gamma^1_{\mathrm{CB}} \geq \gamma_{\mathrm{Th}}|\frac{N_a}{C}-1\big).
\end{align}
When $M$ approaches to infinity, $\gamma^1_{\mathrm{CB}}$ converges to its asymptotic deterministic equivalence $\overline{\gamma}^1_{\mathrm{CB}}$ \cite{33}, which is given by
\begin{align}\label{eq27}
\overline{\gamma}^1_{\mathrm{CB}}=\frac{M}{{N_a}/{C}}v,
\end{align}
where $v=\frac{\rho_{\mathrm{R}}}{1+\rho_{\mathrm{R}}\frac{{N_a}/{C}-1}{{N_a}/{C}}}$. Since $\frac{N_a}{C}$ and $v$ are constant, the value of $\overline{\gamma}^1_{\mathrm{CB}}$ is increased with $M$. When $M$ increases to a certain value, $\gamma^1_{\mathrm{CB}} > \gamma_{\mathrm{Th}}$ and ${P}\big(\gamma^1_{\mathrm{CB}} \geq \gamma_{\mathrm{Th}}|\frac{N_a}{C}-1\big)=1$. Therefore, the success probability approaches to $(1-\frac{1}{P})^{\frac{N_a}{C}-1}$.

We conclude the proof.
\end{proof}

\newtheorem{remark}{Remark}
\begin{remark}
 In a massive MIMO deployment, with large $M$, Proposition $1$ indicates that the success probability of the grant-free RA with massive MIMO mainly depends on $P$, where $P$ is the number of orthogonal preambles.
\end{remark}
\begin{remark}
Based on Proposition $1$, it is clear that the success probability would approach to $1$ when $P$ approaches to infinity and $M$ is sufficiently large.
\end{remark}
\begin{remark}
In the case of $N_a\gg C$, the number of RA UEs distributed on each RA channel with random user distribution would be close to $N_a/C$. Then, the EUD performance of the grant-free RA with massive MIMO could be approximately achieved by the practical case of random user distribution.
\end{remark}
\section{Numerical Results}
In this section, numerical results are presented to verify the accuracy of the analyses in Section III and also the effectiveness of the grant-free RA with massive MIMO. To evaluate the performance, comparisons with three performance baselines are made. The three performance baselines are considered as three different upper bounds of the success probability of the grant-free RA with massive MIMO.
\begin{itemize}
  \item For Upper Bound $1$, EUD is assumed, such that the number of RA UEs distributed on each RA channel is exact $N_a/C$.
  \item For Upper Bound $2$, EUD and $M=\infty$ are assumed. From Proposition $1$, we see that the corresponding success probability is $(1-\frac{1}{P})^{\frac{N_a}{C}-1}$.
  \item For Upper Bound $3$, random user distribution and $P=\infty$ are assumed, i.e., $N_a$ RA UEs are randomly distributed among $C$ RA channels and no preamble collisions occur during the grant-free RA.

\end{itemize}
Although the three upper bounds are impractical considering the randomness of user distribution over channels and finite $P$ and $M$, they would help to understand the effects of these key system parameters on the success probability of the grant-free RA with massive MIMO.
Besides, we define $\eta=N_a/C$ for notation simplicity, which represents the average load on each channel. Please note that, $\eta$ is also a measure of channel reuse efficiency when its corresponding success probability satisfies the system requirement. 
Simulation parameters are summarized in Table \ref{table1}.
\begin{table}[!h]
\renewcommand{\arraystretch}{2}
\caption{Simulation parameters}\label{table1} \centering
\begin{tabular}{>{\centering}m{4cm}|>{\centering}m{4cm}|>{\centering}m{4cm}|}
\hline

\multicolumn{1}{|c|}{Number of antennas $M$}  & \multicolumn{2}{c|}{$50, ~100, ~200, ~400$} \tabularnewline
\hline

\multicolumn{1}{|c|}{{Average load on each channel $\eta$} } & \multicolumn{2}{c|}{{$1 \thicksim 20$} } \\
\hline


\multicolumn{1}{|c|}{Number of orthogonal preambles $P$}  & \multicolumn{2}{c|}{$64, ~128, ~256$} \tabularnewline

\hline

\multicolumn{1}{|c|}{SINR threshold $\gamma_\mathrm{Th}$}  & \multicolumn{2}{c|}{$8$dB} \tabularnewline

\hline
\end{tabular}
\end{table}

\subsection{Success Probability by Conjugate Beamforming under Independent Rayleigh Fading Channel}

The success probabilities as a function of $\rho_\mathrm{R}$ with different values of $M$ and $P$ are presented in Fig. \ref{fig4}, for grant-free RA with CB in massive MIMO, under independent Rayleigh fading channel. As shown in this figure, the analytic results are close to the simulation ones and tighter results are observed as $M$ grows. Moreover, it is clear that the success probability increases as $M$ and $P$ increase.
When $\rho_\mathrm{R}$ increases, the success probability tends to get saturated and the saturated success probability primarily depends on $M$ and $P$.

\begin{figure}[!h]
\centering
\includegraphics[width=3.7in]{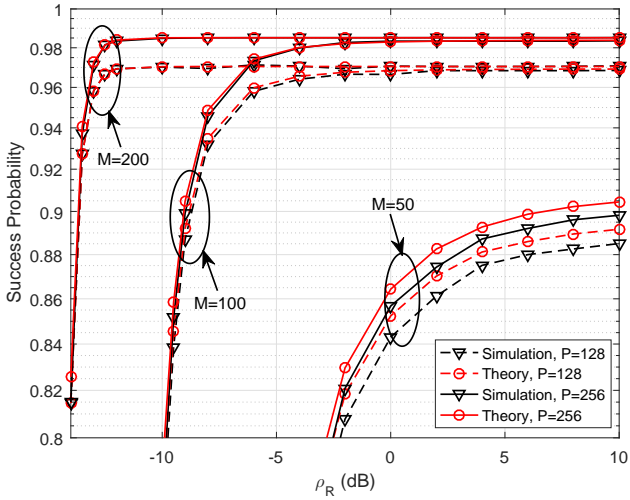}
\caption{Success probability versus $\rho_\mathrm{R}$ with CB under independent Rayleigh fading channel, $\eta=4$.} \label{fig4}
\end{figure}

\begin{figure}[!h]
\centering
\includegraphics[width=3.7in]{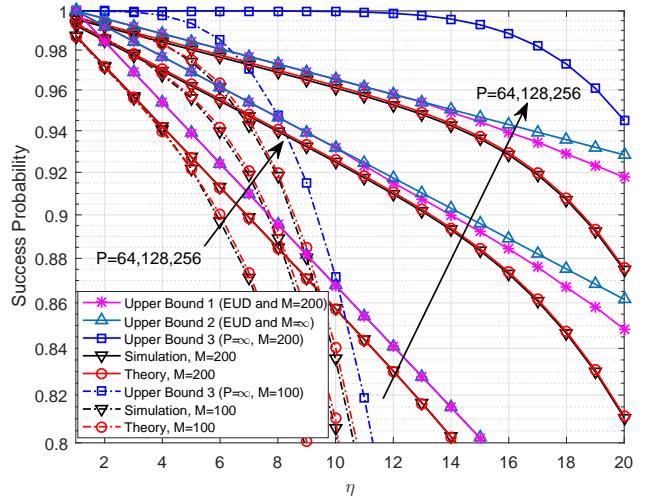}
\caption{Success probability versus $\eta$ with CB under independent Rayleigh fading channel, with different values of $P$ and $M$, $\rho_\mathrm{R}=0$dB.} \label{fig5}
\end{figure}
In Fig. \ref{fig5}, we plot the success probability as a function of $\eta$ at $\rho_\mathrm{R}=0$dB. Various $P$ and $M$ are considered in the figure. As expected, the analytic results closely match with the simulation ones with different $\eta$, $M$, and $P$. Also, it is plain to see that the grant-free RA with CB is capable of simultaneously supporting a number of RA UEs, which is several times the number of RA channels, with high success probability. For example, when $P=256$ and $M=200$, on average $\eta=5$ RA UEs could be accommodated over a same channel with $98\%$ success probability.
As $\eta$ gets larger, we see that the grant-free RA performance is compromised and the success probability declines gradually. It is evident that the performance loss could be compensated by increasing $M$ or $P$.

As observed from Fig. \ref{fig5}, Upper Bound $1$ (the case with the EUD and $M=200$) almost overlaps with Upper Bound $2$ (the case with the EUD and $M=\infty$) with various $P$ for $\eta \leq 14$, which verifies Proposition $1$. Over the range of $5\leq \eta \leq 14$, the success probability is significantly improved by increasing $M$ from $100$ to $200$. As $M$ increases to $200$, it becomes almost parallel with Upper Bound $1$ and Upper Bound $2$ with various $P$. This indicates that $M=200$ is adequate within this range of $\eta$.
Over the range of $2\leq \eta\leq 5$, however, there is little improvement in the success probability with various $P$ when $M$ grows from $100$ to $200$, which shows that $M=100$ is adequate with this range of $\eta$. These observations imply that further increasing $M$ would be of little help to the success probability of the grant-free RA as long as $M \gg \eta$.
In a typical massive MIMO deployment, it is tacit knowledge that the number of antennas is much greater than the number of served UEs over a same channel. Therefore, we conclude that $P$ would dominate the success probability with a wide range of $\eta$ in a typical massive MIMO system, which validates Remark $1$ in Section III. D. In addition, we see that Upper Bound $3$ (the case with $P=\infty$) approaches to $1$ when $M \gg \eta$. This observation confirms Remark $2$. Results also show a small performance gap between the success probability and Upper Bound $1$ with various $P$ for $\eta \leq 14$, which validates Remark $3$.



\subsection{Success Probability by Zero Forcing Beamforming under Independent Rayleigh Fading Channel }

In Fig. \ref{fig7}, success probabilities as a function of $\rho_\mathrm{R}$ with various $M$ and $P$ at $\eta=4$ are illustrated for grant-free RA with ZFB, under independent Rayleigh fading channel. As we can see, the analytic results agree with the simulation ones. Similar to what we observed in Fig. \ref{fig4}, the success probability shows a tendency to saturation with increase of $\rho_\mathrm{R}$. 
Compared to the results with CB in Fig. \ref{fig4}, the success probability with ZFB of $M=50$ almost converges to the one of $M=200$ with different $P$ when $\rho_\mathrm{R}\geq -6$dB, which shows that the ZFB requires less antennas to achieve a given success probability.
Moreover, we see that the grant-free RA with ZFB performs much better at low $\rho_\mathrm{R}$ (e.g., $\rho_\mathrm{R}\leq0$dB) and limited $M$ (e.g., $M=50$). These observations reveal that the ZFB is a better option for the grant-free RA when the receive signals of RA UEs are weak and the number of antennas is limited.
\begin{figure}[!h]
\centering
\includegraphics[width=3.7in]{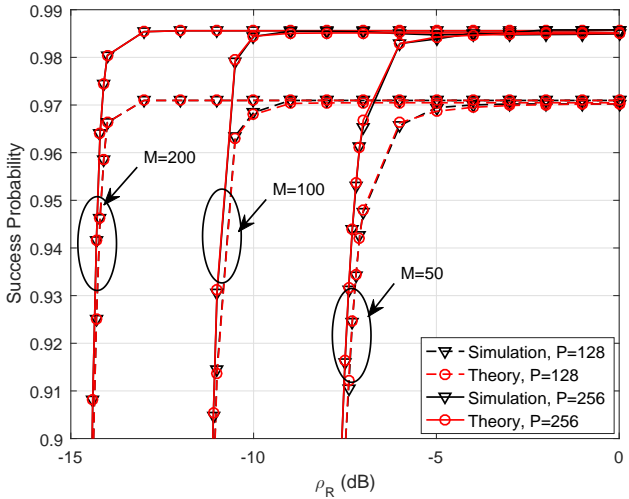}
\caption{Success probability versus $\rho_\mathrm{R}$ with ZFB under independent Rayleigh fading channel, $\eta=4$.} \label{fig7}
\end{figure}

In Fig. \ref{fig8}, we plot the success probability with ZFB as a function of $\eta$ at $\rho_\mathrm{R}=0$dB. We see that the analytic results with ZFB match exactly with the simulation ones under different values of $\eta$, $M$, and $P$. Similar to what we observed in Fig. \ref{fig5}, results validate Proposition $1$ as well as Remark $1-3$. Compared to the case of CB, it is shown that the grant-free RA with ZFB is more effective in the sense that less antennas are required to achieve a given success probability with the same $\eta$. 
In order to accommodate on average $\eta=5$ RA UEs over a same channel with $98\%$ success probability, for instance, much less antennas (e.g., $M=50$) are required by the ZFB, in contrast to $M=200$ by the CB.

\begin{figure}[!h]
\centering
\includegraphics[width=3.7in]{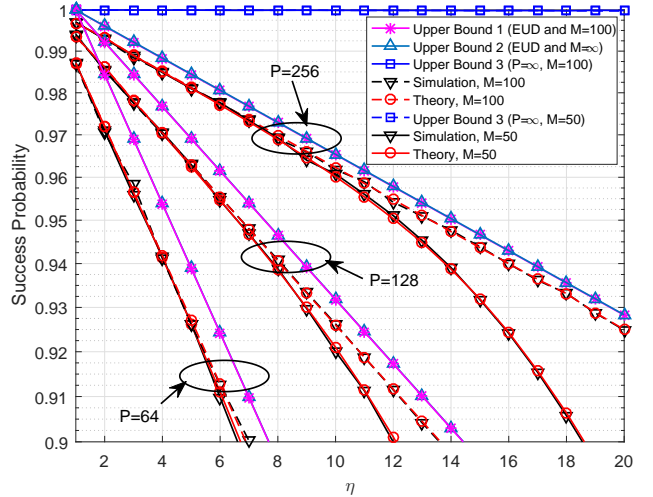}
\caption{Success probability versus $\eta$ with ZFB under independent Rayleigh fading channel, with different values of $P$ and $M$, $\rho_\mathrm{R}=0$dB.} \label{fig8}
\end{figure}

\subsection{Success Probability by Zero-Forcing Beamforming under More Realistic Channel Model }
From Fig. \ref{fig10} to Fig. \ref{fig12}, we consider the success probability of the grant-free RA with ZFB under a more realistic channel model. Specifically, spatially correlated Rayleigh fading channel is considered herein, which has been widely used in MIMO systems for analysis and simulations \cite{40}\cite{41}. The channel response between the BS and an arbitrary RA UE is modelled by
\begin{align*}
\mathbf{g}=\sqrt{\ell}\mathbf{h}=\sqrt{\ell}\mathbf{A}\mathbf{v},
\end{align*}
where $\ell$ denotes large scale fading coefficient between RA UE and BS. $\mathbf{h}=\mathbf{A}\mathbf{v}$ stands for small scale fading vector between RA UE and BS. $\mathbf{A}\in \mathbb{C}^{M\times Q}$ is antenna correlation matrix. $\mathbf{v}\sim \mathcal{CN}(0,\mathbf{I}_{Q})$ is independent fast-fading channel vector, where $Q$ is the number of independently faded paths.

For a uniform linear array, $\mathbf{A}=[\mathbf{a}(\phi_1),\ldots,\mathbf{a}(\phi_Q)]$ is composed of the
steering vector $\mathbf{a}(\phi_q)$ defined as
\begin{align*}
\mathbf{a}(\phi_q)=\frac{1}{\sqrt{Q}}[1,e^{-\textrm{j}2\pi\omega\cos(\phi_q)},\ldots,e^{-\textrm{j}2\pi\omega(M-1)\cos(\phi_q)}]^{\mathrm{T}},
\end{align*}
where $\phi_q$, $q=1,\ldots,Q$, is the angle of arrival (AOA) of the $q$th path, which is uniformly generated within $[\phi_\mathrm{A}-\frac{\phi_\mathrm{S}}{2}, \phi_\mathrm{A}+\frac{\phi_\mathrm{S}}{2}]$.
And $\phi_\mathrm{A}$ and $\phi_\mathrm{S}$ are defined as the azimuth angle of the UE location and the angle spread, respectively. $\omega$ is the antenna spacing in multiples of the wavelength. In practical wireless scenarios, different UEs have different antenna correlation matrix $\mathbf{A}$ due to their random distributions in the cell. Simulation parameters of spatially correlated Rayleigh fading channel are given in Table \ref{table2}.
\begin{table}[!h]
\renewcommand{\arraystretch}{1}
\caption{Simulation Parameters of Spatially Correlated Rayleigh Fading Channel}\label{table2} \centering
\begin{tabular}{>{\centering}m{0.8cm}|>{\centering}m{1.8cm}|>{\centering}m{2cm}|}
\hline

\multicolumn{1}{|c|}{Angle spread $\phi_\mathrm{S}$}  & \multicolumn{2}{c|}{$20^{\circ}$} \tabularnewline
\hline

\multicolumn{1}{|c|}{Azimuth angle $\phi_\mathrm{A}$}  & \multicolumn{2}{c|}{Uniform distribution within $[-60^{\circ},60^{\circ}]$ } \tabularnewline
\hline

\multicolumn{1}{|c|}{Antenna spacing $\omega$}  & \multicolumn{2}{c|}{$1/2$} \tabularnewline
\hline

\multicolumn{1}{|c|}{Number of faded paths $Q$}  & \multicolumn{2}{c|}{${M}/{2}$} \tabularnewline

\hline
\end{tabular}
\end{table}

\begin{figure}[!h]
\centering
\includegraphics[width=3.7in]{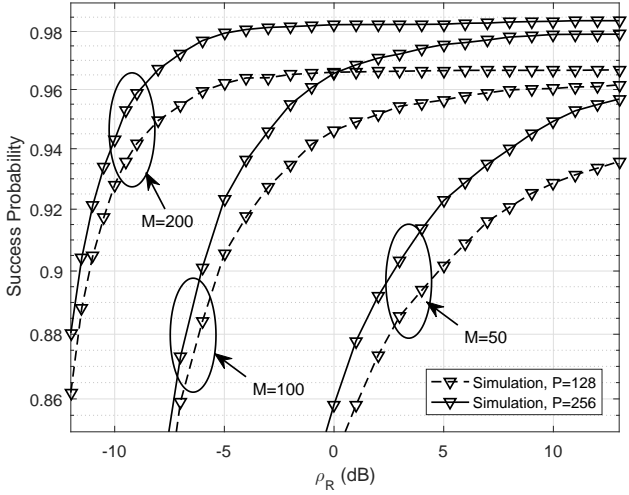}
\caption{Success probability versus $\rho_\mathrm{R}$ with different $M$ under spatially correlated Rayleigh fading channel, $\eta=4$.} \label{fig10}
\end{figure}

\begin{figure}[!h]
\centering
\includegraphics[width=3.7in]{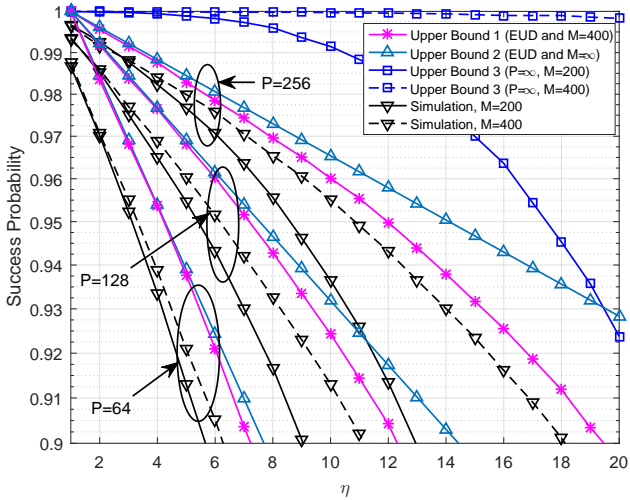}
\caption{Success probability versus $\eta$ with different $P$ under spatially correlated Rayleigh fading channel, $\rho_\mathrm{R}=0$dB.} \label{fig11}
\end{figure}

Comparing Fig. \ref{fig10} and Fig. \ref{fig11} to Fig. \ref{fig7} and Fig. \ref{fig8}, 
 the success probability evidently degrades under the spatially correlated Rayleigh fading channel due to the channel spatial correlations among antennas. Nevertheless, employing more antennas would lessen the effect of channel spatial correlations and compensate the performance loss. Furthermore, it is obersved that Proposition $1$ and Remark $1-3$ is valid under the spatially correlated Rayleigh fading channel, as long as $M$ increases to a sufficiently large value.


\subsection{Merits in terms of MIMO Gain and Gap to EUD}
To further evaluate the merits of the grant-free RA with massive MIMO, comparison between the grant-free RA with massive MIMO and that with single antenna is presented in Fig. \ref{fig12}, Table \ref{table3}, and Table \ref{table4}.

\begin{figure}[!h]
\centering
\includegraphics[width=3.7in]{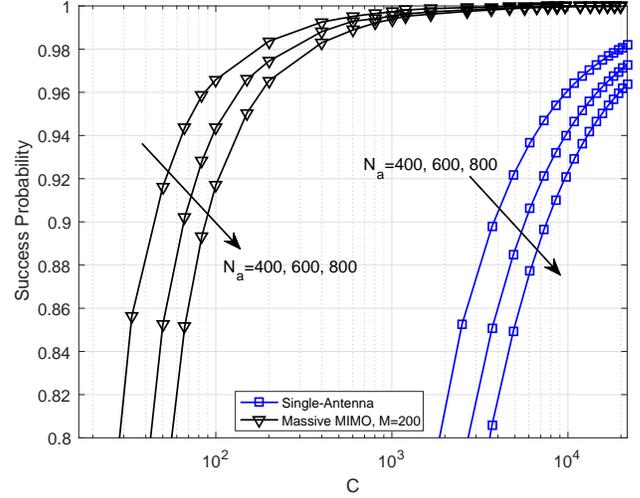}
\caption{Comparison of success probability between the grant-free RA with massive MIMO and that with single antenna, under spatially correlated Rayleigh fading channel. In the case of massive MIMO, we fix $M=200$, $P=128$, and $\rho_\mathrm{R}=0$dB.} \label{fig12}
\end{figure}

 For the single-antenna case, it is equivalent to the slotted ALOHA based random access and collision-free transmissions in the preamble and data domains have to be guaranteed for a successful data recovery. Therefore, the corresponding success probability of an arbitrary RA UE is give by
\begin{align}\label{eq28}
{P}_{\mathrm{SA}}=(1-\frac{1}{C})^{N_a-1},
\end{align}
where the subscript $\mathrm{SA}$ indicates that the single antenna is considered.
\begin{table}[!h]
\renewcommand{\arraystretch}{1}
\caption{MIMO gains with different $M$, under spatially correlated Rayleigh fading channel, $P_\mathrm{ZF}=95\%$, $P=128$, and $\rho_\mathrm{R}=0$dB.}\label{table3} \centering
\begin{tabular}{>{\centering}m{2cm}|>{\centering}m{1cm}|>{\centering}m{1cm}|>{\centering}m{1cm}|>{\centering}m{1cm}|>{\centering}m{1cm}|}
\hline

\multicolumn{1}{|c|}{Number of antennas $M$} & $100$ & $200$ & $400$\tabularnewline
\hline

\multicolumn{1}{|c|}{MIMO gain }  & $35.8$ & $99.6$ & $120.3$ \tabularnewline

\hline
\end{tabular}
\end{table}

As shown in Fig. \ref{fig12}, compared to the single-antenna case, the grant-free RA with massive MIMO is able to achieve a great MIMO gain in terms of $\eta$. Herein, the MIMO gain is given by
\begin{align}\label{eq29}
{\mathrm{GAIN}}_{\mathrm{MIMO}}=\frac{\eta_{M}}{\eta_{S}},
\end{align}
where $\eta_{M}$ is defined as the average number of RA UEs that can be accommodated per RA channel with a targeted success probability in the grant-free RA with massive MIMO. $\eta_{S}$ is defined as the average number of RA UEs that can be accommodated per RA channel with a targeted success probability in the single-antenna case. Obviously, given a $95\%$ success probability and $P=128$, an almost $100$ times MIMO gain is achieved with the simulated $N_a$ by employing $M=200$ antennas in massive MIMO.
In other words, with the targeted success probability of $95\%$ and $P=128$, $200$ antennas is adequate in exchanging for about $100$ times spectrum resources. MIMO gains with different values of $M$ is further shown in Table \ref{table3}.
Considering the fact of spectrum resource scarcity, this rate of resource exchange by massive MIMO could be cost-effective.

To highlight another attractive merit of grant-free RA with massive MIMO, a new term called ``Gap to EUD" (measured in percentage) is coined herein, which is given by
\begin{align}\label{eq30}
{\mathrm{GAP}}_{\mathrm{EUD}}=\frac{\eta_{E}-\eta_{R}}{\eta_{R}},
\end{align}
where $\eta_{E}$ is defined as the average number of RA UEs that can be accommodated per RA channel with a targeted success probability by the grant-free RA with the EUD. $\eta_{R}$ is defined as the average number of RA UEs that can be accommodated per RA channel with a targeted success probability in the case of the random user distribution. The gap to EUD reflects the performance superiority of the EUD over the random user distribution. In the followings, we will show that the gap to EUD vanishes as $M$ increases.

As shown in Fig. \ref{fig11}, given a $95\%$ success probability, when $M=400$ and $P=128$, the gap to EUD of the grant-free RA with ZFB is only about $16\%$, which indicates that the grant-free RA in massive MIMO is able to achieve a close performance to the one of the EUD.
On the contrary, a gap to EUD of about $1917\%$ for the single-antenna case could be observed according to the results in Fig. \ref{fig12}.
\begin{table}[!h]
\renewcommand{\arraystretch}{1}
\caption{Gaps to EUD with different $M$, under spatially correlated Rayleigh fading channel, $P_\mathrm{ZF}=95\%$, $P=128$, and $\rho_\mathrm{R}=0$dB.}\label{table4} \centering
\begin{tabular}{>{\centering}m{1.5cm}|>{\centering}m{0.7cm}|>{\centering}m{0.6cm}|>{\centering}m{0.6cm}|>{\centering}m{0.6cm}|>{\centering}m{0.6cm}|}
\hline

\multicolumn{1}{|c|}{Number of antennas $M$}  & $1$ & $50$ & $100$ & $200$ & $400$\tabularnewline
\hline

\multicolumn{1}{|c|}{Gap to EUD }  & $1917\%$ & $68\%$ & $47\%$ & $24\%$ & $16\%$ \tabularnewline

\hline
\end{tabular}
\end{table}

In Table \ref{table4}, the gaps to EUD with different $M$ are given under the spatially correlated Rayleigh fading channel, where $P_\mathrm{ZF}=95\%$, $P=128$, and $\rho_\mathrm{R}=0$dB. It is evident that increasing $M$ would reduce the gap to EUD and it will be close to $0$ when the number of antennas is massive. As aforementioned, the EUD assumes a genie user distribution management, which is desirable but unattainable in realistic grant-free RA. Fortunately, with massive MIMO, the grant-free RA becomes so effective that the performance of this genie user distribution management is approximately achieved.

\section{Conclusions}
Massive MIMO is a promising technique to greatly increase capacity for future wireless communications.
In this paper, we discussed the success probability of the grant-free RA with massive MIMO and derived the analytic
expressions of success probability for conjugate beamforming and zero-forcing beamforming techniques. Simulation results verified the accuracy of the analyses, and confirmed that the grant-free RA with massive MIMO is capable of supporting $N_a$ simultaneous grant-free access over $C$ RA channels with close-to-one success probability, where $N_a$ is multiple times $C$.
We also showed that, with a specified success probability, utilizing more $M$ and $P$ both provide significant benefits in increasing $\eta$ of the grant-free RA, where $\eta=N_a/C$ reflects the channel reuse efficiency. In other words,
a great MIMO gain in terms of $\eta$ could be achieved for the grant-free RA by massive MIMO compared to its single-antenna counterpart. For instance, given $95\%$ success probability, the grant-free RA with $M=200$ and $P=128$ is able to achieve about $100$ times $\eta$ over the single-antenna counterpart, which is a spectrum saving of about $99\%$.
In addition, the grant-free RA with massive MIMO evidently demonstrates an attractive feature of achieving a close performance to the one with even user distribution (EUD) over channels. This is an important merit of the grant-free RA with massive MIMO, considering the fact that the EUD is desirable but unattainable in the grant-free RA.
Finally, as $M$ in massive MIMO is usually assumed much greater than the average number of served UEs per channel,
we found that $P$ would dominate the success probability within a wide range of $\eta$. Therefore,
preamble designs to increase $P$ and/or reduce preamble collision are very much in need to
ensure the performance gain of the grant-free RA with massive MIMO. 






%
\end{document}